\documentclass[conference]{IEEEtran}

\usepackage{amsmath}
\usepackage{amsfonts}
\usepackage{amssymb}
\usepackage{epsfig}
\usepackage{graphicx}

\newtheorem{theorem}{Theorem}

\newtheorem{corollary}{Corollary}

\begin{document}

\title{Precoding for the AWGN Channel with Discrete Interference}

\author{\authorblockN{Hamid Farmanbar and Amir K. Khandani}
\authorblockA{
Department of Electrical and Computer Engineering\\
University of Waterloo\\
Waterloo, Ontario, Canada N2L 3G1\\
Email: \{hamid,khandani\}@cst.uwaterloo.ca}}

\maketitle

\footnotetext[1]{This work was supported by Nortel, the Natural
Sciences and Engineering Research Council of Canada (NSERC), and the
Ontario Centres of Excellence (OCE).}

\begin{abstract}
For a state-dependent DMC with input alphabet $\mathcal{X}$ and
state alphabet $\mathcal{S}$ where the i.i.d. state sequence is
known causally at the transmitter, it is shown that by using at most
$|\mathcal{X}||\mathcal{S}|-|\mathcal{S}|+1$ out of
$|\mathcal{X}|^{|\mathcal{S}|}$ input symbols of the Shannon's
\emph{associated} channel, the capacity is achievable. As an example
of state-dependent channels with side information at the
transmitter, $M$-ary signal transmission over AWGN channel with
additive $Q$-ary interference where the sequence of i.i.d.
interference symbols is known causally at the transmitter is
considered. For the special case where the Gaussian noise power is
zero, a sufficient condition, which is independent of interference,
is given for the capacity to be $\log_2 M$ bits per channel use. The
problem of maximization of the transmission rate under the
constraint that the channel input given any current interference
symbol is uniformly distributed over the channel input alphabet is
investigated. For this setting, the general structure of a
communication system with optimal precoding is proposed.
\end{abstract}

\section{Introduction}
Information transmission over channels with known interference at
the transmitter has received a great deal of attention. A remarkable
result on such channels was obtained by Costa who showed that the
capacity of the additive white Gaussian noise (AWGN) channel with
additive Gaussian i.i.d. interference, where the sequence of
interference symbols is known non-causally at the transmitter, is
the same as the capacity of AWGN channel \cite{CosIT83}. Therefore,
the interference does not incur any loss in the capacity. This
result was extended to arbitrary interference (random or
deterministic) by Erez \emph{et al.} \cite{Erez05}. The result
obtained by Costa does not hold for the case that the sequence of
interference symbols is known causally at the transmitter.

Channels with known interference at the transmitter are special case
of channels with side information at the transmitter which were
considered by Shannon \cite{Shan58} in causal knowledge setting and
by Gel'fand and Pinsker \cite{Gel80} in non-causal knowledge
setting.

Shannon considered a discrete memoryless channel (DMC) whose
transition matrix depends on the channel state. A state-dependent
discrete memoryless channel (SD-DMC) is defined by a finite input
alphabet $ \mathcal{X}=\{x_1,\ldots,x_{|\mathcal{X}|}\} $, a finite
output alphabet $ \mathcal{Y} $, and transition probabilities $
p(y|x,s) $, where the state $ s $ takes on values in a finite
alphabet $\mathcal{S}=\{1,\ldots,|\mathcal{S}|\}$.

Shannon \cite{Shan58} showed that the capacity of an SD-DMC where
the i.i.d. state sequence is known causally at the encoder is equal
to the capacity of an \emph{associated} regular (without state) DMC
with an extended input alphabet $ \mathcal{T} $ and the same output
alphabet $\mathcal{Y}$. The input alphabet of the associated channel
is the set of all functions from the state alphabet to the input
alphabet of the state-dependent channel. There are a total of $
|\mathcal{X}|^{|\mathcal{S}|} $ of such functions, where $|.|$
denotes the cardinality of a set. Any of the functions can be
represented by a $ |\mathcal{S}| $-tuple $
(x_{i_1},x_{i_2},\ldots,x_{i_{|\mathcal{S}|}}) $ composed of
elements of $\mathcal{X}$, implying that the value of the function
at state $ s $ is $ x_{i_s}, s=1,2,\ldots,|\mathcal{S}| $.

The capacity is given by \cite{Shan58}
\begin{equation} \label{causal-cap}
C = \max_{p(t)} I(T;Y),
\end{equation}
where the maximization is taken over the probability mass function
(pmf) of the random variable $T$.

In the capacity formula (\ref{causal-cap}), we can alternatively
replace $T$ with $(X_1,\ldots,X_{|\mathcal{S}|})$, where $X_s$ is
the random variable that represents the input to the state-dependent
channel when the state is $s, s=1,\ldots,|\mathcal{S}|$.

This paper is organized as follows. In section \ref{bound}, we
derive an upper bound on the cardinality of the Shannon's associated
channel input alphabet to achieve the capacity. In section
\ref{chanmodel}, we introduce our channel model. In section
\ref{nf}, we investigate the capacity of the channel in the absence
of noise. In section \ref{uniform}, we consider maximizing the
transmission rate under the constraint that the channel input given
any current interference symbol is uniformly distributed over the
channel input alphabet. We present the general structure of a
communication system for the channel with causally-known discrete
interference in section \ref{optprecode}. We conclude this paper in
section \ref{conclude}.

\section{A bound on the Cardinality of the Shannon's Associated
Channel input alphabet}                               \label{bound}
We can obtain the pmf of the channel output $Y$ as
\begin{eqnarray} \label{pY} \nonumber
p_Y(y) & = & \sum_{s \in \mathcal{S}} p_S(s) p_{Y|S} (y|s) \\
\nonumber & = & \sum_{s \in \mathcal{S}} p_S(s) \sum_{x \in
\mathcal{X}} p_{X|S}(x|s) p_{Y|X,S} (y|x,s)\\
& = & \sum_{s \in \mathcal{S}} p_S(s) \sum_{x \in \mathcal{X}}
p_{X_s}(x) p_{Y|X,S} (y|x,s).
\end{eqnarray}
The capacity of the associated channel, which is the same as the
capacity of the original state-dependent channel, is the maximum of
$ I(T;Y) = I(X_1 X_2 \cdots X_{|\mathcal{S}|};Y)$ over the joint pmf
values $ p_{i_1 i_2 \cdots i_{|\mathcal{S}|}} = \mbox{Pr} \{X_1 =
x_{i_1}, \ldots, X_{|\mathcal{S}|} = x_{i_{|\mathcal{S}|}}\}$, i.e.,
\begin{equation} \label{capg}
C = \max_{p_{i_1 i_2 \cdots i_{|\mathcal{S}|}}} \hspace{5pt}I(X_1
X_2 \cdots X_{|\mathcal{S}|};Y).
\end{equation}
The mutual information between $ T $ and $ Y $ is the difference
between the entropies $ H(Y) $ and $ H(Y|T) $. It can be seen from
(\ref{pY}) that $p_Y(y)$, and hence $H(Y)$, are uniquely determined
by the marginal pmfs $\{p_{X_{s}}(x_i)\}_{i=1}^{|\mathcal{X}|}$, $
s=1,\ldots, |\mathcal{S}| $. The conditional entropy $H(Y|T)$ is
given by
\begin{eqnarray} \nonumber
H(Y|T) \hspace{-8pt} & = & \hspace{-8pt} H(Y|X_1 X_2 \cdots
X_{|\mathcal{S}|}) \\ \hspace{-8pt} & = & \hspace{-8pt}
\sum_{i_1=1}^{|\mathcal{X}|} \cdots
\sum_{i_{|\mathcal{S}|}=1}^{|\mathcal{X}|} h_{i_1 \cdots
i_{|\mathcal{S}|}} p_{i_1 \cdots i_{|\mathcal{S}|}},
\end{eqnarray}
where $h_{i_1 \cdots i_{|\mathcal{S}|}} = H(Y|X_1 = x_{i_1},\ldots,
X_{|\mathcal{S}|}=x_{i_{|\mathcal{S}|}})$.

There are $|\mathcal{X}|^{|\mathcal{S}|}$ variables involved in the
maximization problem (\ref{capg}). Each variable represents the
probability of an input symbol of the associated channel. The
following theorem regards the number of nonzero variables required
to achieve the maximum in (\ref{capg}).
\begin{theorem} \label{th1}
The capacity of the associated channel is achieved by using at most
$|\mathcal{X}||\mathcal{S}|-|\mathcal{S}|+1$ out of
$|\mathcal{X}|^{|\mathcal{S}|}$ input symbols with nonzero
probabilities.
\end{theorem}
\begin{proof}
Denote by $\{\hat{p}_{i}^{(s)}\}_{i=1}^{|\mathcal{X}|}
=\{\hat{p}_{X_s}(x_i)\}_{i=1}^{|\mathcal{X}|}$ the pmf of $X_s$,
$s=1,2,\ldots,|\mathcal{S}|$, induced by a capacity-achieving joint
pmf $\{\hat{p}_{i_1 \cdots i_{|\mathcal{S}|}}\}_{i_1, \ldots,
i_{|\mathcal{S}|}=1}^{|\mathcal{X}|}$. We limit the search for a
capacity-achieving joint pmf to those joint pmfs that yield the same
marginal pmfs as $\{\hat{p}_{i_1 \cdots i_{|\mathcal{S}|}}\}_{i_1,
\ldots, i_{|\mathcal{S}|=1}}^{|\mathcal{X}|}$. By limiting the
search to this smaller set, the maximum of $I(X_1\cdots
X_{|\mathcal{S}|};Y)$ remains unchanged since the capacity-achieving
joint pmf $\{\hat{p}_{i_1 \cdots i_{|\mathcal{S}|}}\}_{i_1, \ldots,
i_{|\mathcal{S}|}=1}^{|\mathcal{X}|}$ is in the smaller set. But all
joint pmfs in the smaller set yield the same $H(Y)$ since they
induce the same marginal pmfs on $X_1,\ldots,X_{|\mathcal{S}|}$.
Therefore, the maximization problem in (\ref{capg}) reduces to the
linear minimization problem
\begin{eqnarray} \nonumber \label{LP}
\min_{p_{i_1 \cdots i_{|\mathcal{S}|}}} \hspace{-8pt}& &
\hspace{-8pt} \sum_{i_1=1}^{|\mathcal{X}|} \cdots
\sum_{i_{|\mathcal{S}|}=1}^{|\mathcal{X}|} h_{i_1\cdots
i_{|\mathcal{S}|}} p_{i_1 \cdots i_{|\mathcal{S}|}}  \\ \nonumber
\mbox{s. t.} & & \\ \nonumber \hspace{-8pt} & &
\hspace{-8pt}\sum_{i_2=1}^{|\mathcal{X}|} \cdots
\sum_{i_{|\mathcal{S}|}=1}^{|\mathcal{X}|} p_{i_1 \cdots
i_{|\mathcal{S}|}} = \hat{p}_{i_1}^{(1)}, \hspace{5pt}
i_1=1,\ldots,|\mathcal{X}|, \\  \nonumber \hspace{-8pt} & &
\hspace{17pt} \vdots \hspace{150pt} \vdots \\ \nonumber
\hspace{-8pt} & & \hspace{-8pt} \sum_{i_1=1}^{|\mathcal{X}|} \cdots
\sum_{i_{{|\mathcal{S}|}-1}=1}^{|\mathcal{X}|} p_{i_1 \cdots
i_{|\mathcal{S}|}} =
\hat{p}_{i_{|\mathcal{S}|}}^{({|\mathcal{S}|})},
\hspace{5pt} i_{|\mathcal{S}|}=1,\ldots,|\mathcal{X}|,\\
\hspace{-8pt}& & \hspace{-8pt} p_{i_1 \cdots i_{|\mathcal{S}|}} \geq
0, \hspace{10pt} i_1,
\ldots,i_{|\mathcal{S}|}=1,2,\ldots,|\mathcal{X}|.
\end{eqnarray}
There are $|\mathcal{X}||\mathcal{S}|$ equality constraints in
(\ref{LP}) out of which $|\mathcal{X}||\mathcal{S}|-|\mathcal{S}|+1$
are linearly independent. From the theory of linear programming, the
minimum of (\ref{LP}), and hence the maximum of $I(X_1 \cdots
X_{|\mathcal{S}|};Y)$, is achieved by a feasible solution with at
most $|\mathcal{X}||\mathcal{S}|-|\mathcal{S}|+1$ nonzero variables.
\end{proof}
Theorem \ref{th1} states that at most
$|\mathcal{X}||\mathcal{S}|-|\mathcal{S}|+1$ out of
$|\mathcal{X}|^{|\mathcal{S}|}$ input symbols of the associated
channel are needed to be used with positive probability to achieve
the capacity. However, in general one does not know which of the
inputs must be used to achieve the capacity. If we knew the marginal
pmfs for $X_1,\ldots,X_{|\mathcal{S}|}$ induced by a
capacity-achieving joint pmf, we could obtain the capacity-achieving
joint pmf itself by solving the linear program (\ref{LP}).

\section{The Channel Model} \label{chanmodel}
We consider data transmission over the channel
\begin{equation} \label{chmodel}
Y = X + S + N,
\end{equation}
where $ X $ is the channel input, which takes on values in a fixed
real constellation
\begin{equation}
\mathcal{X} = \left \{ x_1, x_2,\ldots, x_M \right\},
\end{equation}
$ Y $ is the channel output, $ N $ is additive white Gaussian noise
with power $ P_N $, and the interference $ S $ is a discrete random
variable that takes on values in
\begin{equation}
\mathcal{S}=\left\lbrace s_1, s_2, \ldots, s_Q \right\rbrace
\end{equation}
with probabilities $r_1,r_2, \ldots, r_Q$, respectively. The
sequence of i.i.d. interference symbols is known causally at the
encoder. The above channel can be considered as a special case of
state-dependent channels considered by Shannon with one exception,
that the channel output alphabet is continuous. In our case, the
likelihood function $ f_{Y|X,S}(y|x,s) $ is used instead of the
transition probabilities. We denote the input to the associated
channel by $ T $, which can also be represented as $
(X_1,X_2,\ldots,X_Q) $, where $ X_j $ is the random variable that
represents the channel input when the current interference symbol is
$s_j$, $j=1,\ldots,Q$.

The likelihood function for the associated channel is given by
\begin{eqnarray} \nonumber
f_{Y|T}(y|t) & = & \sum_{j=1}^{Q} r_j f_{Y|X,S}(y|x_{i_{j}},s_j)\\
& = & \sum_{j=1}^{Q} r_j f_{N}(y-x_{i_{j}}-s_j),
\end{eqnarray}
where $ f_N $ denotes the pdf of the noise $ N $, and $t$ is the
input symbol of the associated channel represented by
$(x_{i_{1}},x_{i_{2}}, \ldots, x_{i_{Q}})$.

According to theorem \ref{th1}, the capacity of our channel is
obtained by using at most $MQ-Q+1$ out of $M^Q$ input symbols of the
associated channel.

\section{The Noise-Free Channel} \label{nf}
We consider a special case where the noise power is zero in
(\ref{chmodel}). In the absence of noise, the channel output $Y$
takes on at most $MQ$ different values since different $X$ and $S$
pairs may yield the same sum. If $Y$ takes on exactly $MQ$ different
values, then it is easy to see that the capacity is $\log_2 M$ bits
\footnote{This is true even if the interference sequence is unknown
to the encoder.}: The decoder just needs to partition the set of all
possible channel output values into $M$ subsets of size $Q$
corresponding to $M$ possible inputs, and decide that which subset
the current received symbol belongs to.

In general, where the cardinality of the channel output symbols can
be less than $MQ$, we will show that under some condition on the
channel input alphabet there exists a coding scheme that achieves
the rate $\log_2 M$ in one use of the channel. We do this by
considering a one-shot coding scheme which uses only $M$ (out of
$M^Q$) inputs of the associated channel.

In a one-shot coding scheme, a message is encoded to a single input
of the associated channel. Any input of the associated channel can
be represented by a $Q$-tuple composed of elements of $\mathcal{X}$.
Given that the current interference symbol is $s_j$, the $j$th
element of the $Q$-tuple is sent through the channel. Therefore, one
single message can result in (up to) $Q$ symbols at the output. For
convenience, we consider the output symbols corresponding to a
single message as a multi-set\footnote{A multi-set differs from a
set in that each member may have a multiplicity greater than one.
For example, $\{1,3,3,7\}$ is a multi-set of size four where 3 has
multiplicity two.} of size (exactly) $Q$. If the $M$ multi-sets at
the output corresponding to $M$ different messages are mutually
disjoint, reliable transmission through the channel is possible.

Unfortunately, we cannot always find $M$ inputs of the associated
channel such that the corresponding multi-sets are mutually
disjoint. For example, consider a channel with the input alphabet
$\mathcal{X}=\{0,1,2,4\}$ and the interference alphabet
$\mathcal{S}=\{0,1,3\}$. It is easy to check that for this channel
we cannot find four triples composed of elements of $\mathcal{X}$
such that the corresponding multi-sets are mutually disjoint. In
fact, by entropy calculations we can show that the capacity of the
channel in this example is less than $2$ bits.

However, if we put some constraint on the channel input alphabet,
the rate $\log_2 M$ is achievable.

\begin{theorem} \label{th2}
Suppose that the elements of the channel input alphabet
$\mathcal{X}$ form an arithmetic progression. Then the capacity of
the noise-free channel
\begin{equation} \label{noise-free}
Y = X + S,
\end{equation}
where the sequence of interference symbols is known causally at the
encoder equals $\log_2 M$ bits.
\end{theorem}

\begin{proof}
Let $\mathcal{Y}^{(q)}$ be the set of all possible outputs of the
noise-free channel when the interference symbol is $s_q$, i.e.,
\begin{equation}
\mathcal{Y}^{(q)} = \left\{x_1+s_q, x_2+s_q, \ldots,
x_M+s_q\right\}, \quad q = 1, \ldots, Q.
\end{equation}
The union of $\mathcal{Y}^{(q)}$s is the set of all possible outputs
of the noise-free channel.

Without loss of generality we can assume that $s_1 < s_2 < \cdots
<s_Q$. The elements of $\mathcal{Y}^{(q)}$ form an arithmetic
progression, $q = 1, \ldots, Q$. Furthermore, these $Q$ arithmetic
progressions are shifted versions of each other.

We prove by induction on $Q$ that there exist $M$ mutually-disjoint
multi-sets of size $Q$ composed of the elements of
$\mathcal{Y}^{(1)}, \mathcal{Y}^{(2)}, \ldots, \mathcal{Y}^{(Q)}$
(one element from each). If we can find such $M$ multi-sets of size
$Q$, then we can obtain the corresponding $M$ $Q$-tuples of elements
of $\mathcal{X}$ by subtracting the corresponding interference terms
from the elements of the multi-sets. These $M$ $Q$-tuples can serve
as the inputs of the associated channel to be used for sending any
of $M$ distinct messages through the channel without error in one
use of the channel, hence achieving the rate $\log_2 M$ bits per
channel use.

For $Q=1$, the statement of the theorem is true since we can take
$\{x_1+s_1 \}, \{x_2+s_1 \}, \ldots, \{x_M+s_1 \}$ as
mutually-disjoint sets of size one.

Assume that there exist $M$ mutually-disjoint multi-sets of size
$Q=q$. For $Q=q+1$, we will have the new set of channel outputs
$\mathcal{Y}^{(q+1)}=\{x_1+s_{q+1}, x_2+s_{q+1}, \ldots, x_M+s_{q+1}
\}$. We consider two possible cases:

\textit{Case 1}: None of the elements of $\mathcal{Y}^{(q +1)}$
appear in any of the multi-sets of size $Q=q$.

In this case, we include the elements of $\mathcal{Y}^{(q+1)}$ in
the $M$ multi-sets  arbitrarily (one element is included in each
multi-set). It is obvious that the resulting multi-sets of size
$Q=q+1$ are mutually disjoint.

\textit{Case 2}: Some of the elements of $\mathcal{Y}^{(q +1)}$
appear in some of the multi-sets of size $Q=q$.

Suppose that the largest element of $\mathcal{Y}^{(q +1)}$ which
appears in any of the sets $\mathcal{Y}^{(1)},\ldots$,
$\mathcal{Y}^{(q)}$ (or equivalently, in any of the multi-sets of
size $Q=q$) is $x_k + s_{q+1}$ for some $1 \leq k \leq M-1$. Then
since $\mathcal{Y}^{(q +1)}$ is shifted version of each
$\mathcal{Y}^{(1)},\ldots, \mathcal{Y}^{(q)}$ and $s_{q+1} > s_q >
\cdots > s_1$, exactly one of the sets $\mathcal{Y}^{(1)},\ldots,
\mathcal{Y}^{(q)}$, say $\mathcal{Y}^{(j)}$ for some $1 \leq j \leq
q$, contains all elements of $\mathcal{Y}^{(q +1)}$ up to
$x_k+s_{q+1}$. See fig. \ref{demo}. Since any of the disjoint
multi-sets of size $Q$ contain just one element of
$\mathcal{Y}^{(j)}$, the elements of $\mathcal{Y}^{(q +1)}$ up to
$x_k+s_{q+1}$ appear in different multi-sets of size $Q=q$. We can
form the disjoint multi-sets of size $q+1$ by including these common
elements in the corresponding multi-sets and including the elements
of $\{x_{k+1}+s_{q+1}, \ldots, x_M+s_{q+1} \}$ in the remaining
multi-sets arbitrarily.
\end{proof}
\begin{figure}[t]
\centering
\includegraphics[scale = 0.55]{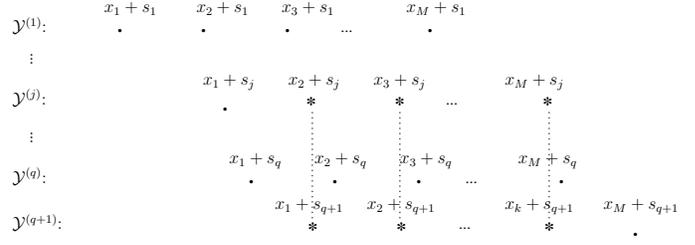}
\caption{The elements of $\mathcal{Y}^{(1)},\ldots,
\mathcal{Y}^{(q+1)}$ shown as shifted version of each other. The
elements of $\mathcal{Y}^{(q+1)}$ up to $x_k+s_{q+1}$ appear in
$\mathcal{Y}^{(j)}$.} \label{demo}
\end{figure}

The condition on the channel input alphabet in the statement of
theorem \ref{th2} is a sufficient condition for the channel capacity
to be $\log_2 M$. However, it is not a necessary condition. For
example, the statement of theorem \ref{th2} without that condition
is true for the case of $Q=2$. Because in the second iteration, we
do not need the arithmetic progression condition to form $M$
mutually-disjoint multi-sets of size two.

The proof of theorem \ref{th2} is actually a constructive algorithm
for finding $M$ (out of $M^Q$) inputs of the associated channel to
be used with probability $\frac{1}{M}$ to achieve the rate $\log_2
M$ bits.

It is interesting to see that the set containing the $q$th elements
of the $M$ $Q$-tuples obtained by the constructive algorithm is
$\mathcal{X}$, $q=1,\ldots,Q$. This is due to the fact that each
multi-set contains one element from each $\mathcal{Y}^{(1)},\ldots,
\mathcal{Y}^{(Q)}$. Therefore, a uniform distribution on the $M$
$Q$-tuples induces uniform distribution on $X_1, \ldots,X_Q$.

\section{Uniform Transmission}\label{uniform}
In the sequel, we study the maximization of the rate $I(X_1\cdots
X_Q;Y)$ over joint pmfs $\{p_{i_1 \cdots i_Q} \}_{i_1,\ldots,i_Q =
1}^M$ that induce uniform marginal distributions on
$X_1,\ldots$,$X_Q$, i.e.,
\begin{equation} \label{margin}
p_{i}^{(1)} = p_{i}^{(2)} = \cdots = p_{i}^{(Q)} = \frac{1}{M},
\hspace{20pt} i = 1, 2, \ldots, M,
\end{equation}
for which we show how to obtain the optimal input probability
assignment. We call a transmission scheme that induces uniform
distribution on $X_1, \ldots,X_Q$ as \emph{uniform transmission}.
The uniform distribution for $X_1, \ldots,X_Q$ implies uniform
distribution for $X$, the input to the state-dependent channel
defined in (\ref{chmodel}).

In the previous section, we established that the capacity achieving
pmf for the asymptotic case of noise-free channel induces uniform
distributions on $X_1, \ldots,X_Q$ (provided that we can find $M$
$Q$-tuples such that the corresponding multi-sets are mutually
disjoint).

Considering the constraints in (\ref{margin}), the maximization of
$I(X_1\cdots X_Q;Y)$ is reduced to the linear minimization problem
\begin{eqnarray} \nonumber \label{LPUN}
\min_{p_{i_1 \cdots i_Q}}& & \sum_{i_1=1}^{M} \cdots
\sum_{i_Q=1}^{M} h_{i_1\cdots i_Q} p_{i_1 \cdots i_Q}  \\ \nonumber
\mbox{s. t.} & &                                  \\ \nonumber & &
\sum_{i_2=1}^{M} \cdots \sum_{i_Q=1}^{M} p_{i_1 \cdots i_Q}  =
\frac{1}{M}, \hspace{25pt} i_1=1,\ldots,M,  \\
\nonumber & & \hspace{25pt} \vdots \hspace{150pt} \vdots         \\
\nonumber & & \sum_{i_1=1}^{M} \cdots \sum_{i_{Q-1}=1}^{M} p_{i_1
\cdots i_Q} = \frac{1}{M}, \hspace{12pt} i_Q=1,\ldots,M,\\
& & p_{i_1 \cdots i_Q} \geq 0, \hspace{25pt} i_1,
\ldots,i_Q=1,2,\ldots,M.
\end{eqnarray}
The same argument used in the last part of the proof of theorem
\ref{th1} can be used to show that the maximum is achieved by using
at most $MQ-Q+1$ inputs of the associated channel with positive
probabilities. This is restated in the following corollary.

\begin{corollary} \label{coro1}
The maximum of $I(X_1 \cdots X_Q;Y)$ over joint pmfs $\{p_{i_1
\cdots i_Q} \}_{i_1,\ldots,i_Q = 1}^M$ that induce uniform marginal
distributions on $X_1,X_2,\ldots,X_Q$ is achieved by a joint pmf
with at most $MQ-Q+1$ nonzero elements.
\end{corollary}

This result is independent of the coefficients $\{h_{i_1\cdots
i_Q}\}$. However, which probability assignment with at most $MQ-Q+1$
nonzero elements is optimal depends on the coefficients $\{ h_{i_1
\cdots i_Q} \} $. The coefficient $h_{i_1 \cdots i_Q}$ is determined
by the interference levels $s_1, \ldots, s_Q$, the probability of
interference levels $r_1, \ldots, r_Q$, the noise power $P_N$, and
the signal points $x_1,x_2, \ldots, x_M$. The optimal probability
assignment is obtained by solving the linear programming problem
(\ref{LPUN}) using the simplex method \cite{Kreko}.

\subsection{Two-Level Interference}
If the number of interference levels is two, i.e., $Q=2$, we can
make a stronger statement than corollary \ref{coro1}.
\begin{theorem} \label{th3}
The maximum of $ I(X_1 X_2;Y) $ over $
\{p_{i_1i_2}\}_{i_1,i_2=1}^{M} $ with uniform marginal pmfs for
$X_1$ and $X_2$ is achieved by using exactly $M$ out of $M^2$ inputs
of the associated channel with probability $\frac{1}{M}$.
\end{theorem}
\begin{proof}
The equality constraints of (\ref{LPUN}) can be written in matrix
form as
\begin{equation} \label{matrix}
\mathbf{A}\mathbf{p} = \mathbf{1},
\end{equation}
where $\mathbf{A}$ is a zero-one $MQ \times M^Q$ matrix,
$\mathbf{p}$ is $M$ times the vector containing all $p_{i_1 \cdots
i_Q}$s in lexicographical order, and $\mathbf{1}$ is the all-one $MQ
\times 1$ vector.

For $Q=2$, it is easy to check that $\mathbf{A}$ is the vertex-edge
incidence matrix of $K_{M,M}$, the complete bipartite graph with $M$
vertices at each part. Therefore, $\mathbf{A}$ is a totally
unimodular matrix\footnote{A totally unimodular matrix is a matrix
for which every square submatrix has determinant $0$,$1$, or $-1$.}
\cite{Nemhauser88}. Hence, the extreme points of the feasible region
$ F = \left\lbrace \mathbf{p}: \mathbf{Ap}=\mathbf{1}, \mathbf{p}
\geq \mathbf{0} \right\rbrace $ are integer vectors. Since the
optimal value of a linear optimization problem is attained at one of
the extreme points of its feasible region, the minimum in
(\ref{LPUN}) is achieved at an all-integer vector $\mathbf{p}^*$.
Considering that $\mathbf{p}^*$ satisfies (\ref{matrix}), it can
only be a zero-one vector with exactly $M$ ones.
\end{proof}

\begin{figure}[t]
\centering
\includegraphics[scale = 0.5]{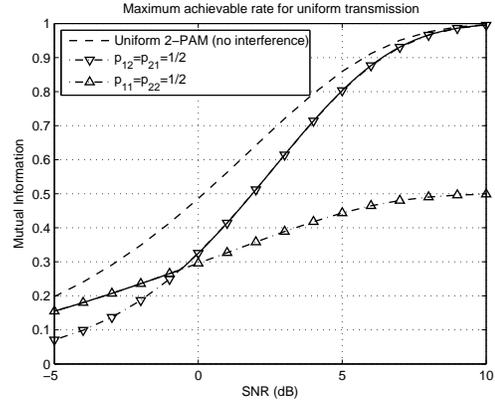}
\caption{Maximum mutual information vs. SNR for the channel with
$\mathcal{X} = \mathcal{S} = \{-1,+1\}$ and $r_1 = r_2 =
\frac{1}{2}$.} \label{MI1}
\end{figure}

Fig. \ref{MI1} depicts the maximum mutual information (for the
uniform transmission scenario) vs. SNR for the channel with
$\mathcal{X} = \mathcal{S} = \{-1,+1\}$ and equiprobable
interference symbols. The mutual information vs. SNR curve for the
interference-free AWGN channel with equiprobable input alphabet
$\{-1,+1\}$ is plotted for comparison purposes. As it can be seen,
for low SNRs, the input probability assignment $p_{11} = p_{22} =
\frac{1}{2}$ is optimal, whereas at high SNRs, the input probability
assignment $p_{12} = p_{21} = \frac{1}{2}$ is optimal. The maximum
achievable rate for uniform transmission is the upper envelope of
the two curves corresponding to different input probability
assignments. Also, it can be observed that the achievable rate
approaches $\log_2 2 = 1$ bit per channel use as SNR increases
complying with the fact that we established in section \ref{nf} for
the noise-free channel.

It turns out from the proof of theorem \ref{th3} that the optimum
solution of the linear optimization problem, $\mathbf{p}^*$, is a
zero-one vector. So, if we add the integrality constraint to the set
of constraints in (\ref{matrix}), we still obtain the same optimal
solution. The resulting integer linear optimization problem is
called the \emph{assignment problem} \cite{Nemhauser88}, which can
be solved using low-complexity algorithms such as the
\emph{Hungarian method} \cite{Kreko}.

\subsection{Integrality Constraint for the $Q$-Level Interference}
The fact that for the case $Q=2$, there exists an optimal
$\mathbf{p}$ which is a zero-one vector with exactly $M$ ones
simplifies the encoding operation. Because any encoding scheme just
needs to work on a subset of size $M$ of the associated channel
input alphabet with equal probabilities $\frac{1}{M}$.

For $Q\neq2$, $\mathbf{A}$ is not a totally unimodular matrix.
Therefore, not all extreme points of the feasible region defined by
$\mathbf{A}\mathbf{p} = \mathbf{1}, \mathbf{p} \geq \mathbf{0}$, are
integer vectors. However, at the expense of possible loss in rate,
we may add the integrality constraint in this case. The resulting
optimization problem is called the \emph{multi-dimensional
assignment problem} \cite{MAP}. The optimal solution of (\ref{LPUN})
with the integrality constraint, will be a vector with exactly $M$
nonzero elements with the value $\frac{1}{M}$. Therefore, any
encoding scheme just needs to use $M$ symbols of the associated
channel with equal probabilities, simplifying the encoding
operation.

\section{Optimal Precoding} \label{optprecode}
The general structure of a communication system for the channel
defined in (\ref{chmodel}) is shown in fig. \ref{fig6}. Any encoding
and decoding scheme for the associated channel can be translated to
an encoding and decoding scheme for the original channel defined in
(\ref{chmodel}). A message $w$ is encoded into a block of length $n$
composed of input symbols of the associated channel $t \sim
(x_{i_1},x_{i_2},\ldots,x_{i_Q})$. There are $M^Q$ input symbols.
However, we showed that the maximum rate with uniformity and
integrality constraints can be achieved by using just $M$ input
symbols of the associated channel with equal probabilities. The
optimal $M$ input symbols of the associated channel are obtained by
solving the linear programming problem (\ref{LPUN}) with the
integrality constraint. Those $M$ input symbols of the associated
channel define the optimal precoding operation: For any $t$ that
belongs to the set of $M$ optimal input symbols, the precoder sends
the $q$th component of $t$ if the current interference symbol is
$s_q$, $q=1,\ldots,Q$. Based on the received sequence, the receiver
decodes $\hat{w}$ as the transmitted message.
\begin{figure}[tbp]
\centering
\includegraphics[scale=0.7]{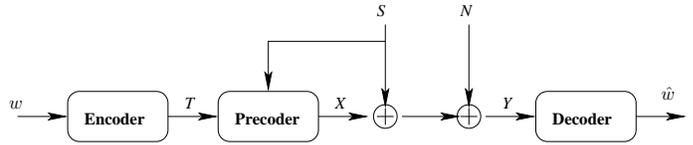}
\caption{General structure of the communication system for channels
with causally-known discrete interference.} \label{fig6}
\end{figure}

\section{Conclusion} \label{conclude}
In this paper, we proved that the capacity of an SD-DMC with finite
input alphabet $\mathcal{X}$ and finite state alphabet $\mathcal{S}$
and with causally known i.i.d. state sequence at the encoder can be
achieved by using at most
$|\mathcal{X}||\mathcal{S}|-|\mathcal{S}|+1$ out of
$|\mathcal{X}|^{|\mathcal{S}|}$ input symbols of the associated
channel. As an example of state-dependent channels with side
information at the encoder, we investigated $M$-ary signal
transmission over AWGN channel with additive $Q$-level interference,
where the sequence of interference symbols is known causally at the
transmitter.

For the noise-free channel, provided that the signal points are
equally spaced, we proposed a one-shot coding scheme that uses $M$
input symbols of the associated channel to achieves the capacity
$\log_2 M$ bits.

We considered the transmission schemes with uniform pmfs for
$X_1,\ldots,X_Q$. For this so called uniform transmission, the
optimal input probability assignment with at most $MQ-Q+1$ nonzero
elements can be obtained by solving the linear optimization problem
(\ref{LPUN}). The optimal solution to (\ref{LPUN}) with the
integrality constraint has exactly $M$ nonzero elements. For the
case $Q=2$, we showed that the integrality constraint does not
reduce the maximum achievable rate. The loss in rate (if there is
any) by imposing the integrality constraint for the general case is
a problem to be explored.

\section*{Acknowledgment}
The authors would like to thank Gerhard Kramer and Syed Ali Jafar
for pointing out that the proof of theorem \ref{th1}, which was
originally given for the channel model considered in this paper,
works for any SD-DMC with causal side information at the encoder.

\end{document}